\documentclass[11pt,a4paper]{article}

\usepackage{amsmath}
\usepackage{amssymb}
\usepackage{url}
\usepackage{hyperref}

\usepackage{graphicx}
\usepackage[margin=1.2in]{geometry}

\usepackage{breakurl}
\RequirePackage[capitalize]{cleveref}[0.19]
\crefname{section}{section}{sections}
\crefname{subsection}{subsection}{subsections}
\Crefname{figure}{Figure}{Figures}
\crefformat{equation}{\textup{#2(#1)#3}}
\crefrangeformat{equation}{\textup{#3(#1)#4--#5(#2)#6}}
\crefmultiformat{equation}{\textup{#2(#1)#3}}{ and \textup{#2(#1)#3}}
{, \textup{#2(#1)#3}}{, and \textup{#2(#1)#3}}
\crefrangemultiformat{equation}{\textup{#3(#1)#4--#5(#2)#6}}%
{ and \textup{#3(#1)#4--#5(#2)#6}}{, \textup{#3(#1)#4--#5(#2)#6}}{, and \textup{#3(#1)#4--#5(#2)#6}}
\Crefformat{equation}{#2Equation~\textup{(#1)}#3}
\Crefrangeformat{equation}{Equations~\textup{#3(#1)#4--#5(#2)#6}}
\Crefmultiformat{equation}{Equations~\textup{#2(#1)#3}}{ and \textup{#2(#1)#3}}
{, \textup{#2(#1)#3}}{, and \textup{#2(#1)#3}}
\Crefrangemultiformat{equation}{Equations~\textup{#3(#1)#4--#5(#2)#6}}%
{ and \textup{#3(#1)#4--#5(#2)#6}}{, \textup{#3(#1)#4--#5(#2)#6}}{, and \textup{#3(#1)#4--#5(#2)#6}}
\crefdefaultlabelformat{#2\textup{#1}#3}

\usepackage{amsthm}
\newtheorem{theorem}{Theorem}

\begin{document}

\title{Computing hypergeometric functions rigorously}

%\author{Fredrik Johansson\thanks{LFANT, INRIA and Institut de Math\'{e}matiques de Bordeaux
%  (\protect\url{fredrik.johansson@gmail.com}}). }

\author{Fredrik Johansson\footnote{LFANT, INRIA and Institut de Math\'{e}matiques de Bordeaux (\protect\url{fredrik.johansson@gmail.com}).
  This research was partially funded by ERC Starting Grant ANTICS 278537.}}

\date{}

\maketitle

\begin{abstract}
We present an efficient implementation of hypergeometric functions
in arbitrary-precision interval arithmetic.
The functions ${}_0F_1$, ${}_1F_1$, ${}_2F_1$ and ${}_2F_0$
(or the Kummer $U$-function) are supported
for unrestricted complex parameters and argument,
and by extension,
we cover
exponential and trigonometric integrals, error functions, Fresnel integrals,
incomplete gamma and beta functions, Bessel functions, Airy functions,
Legendre functions, Jacobi polynomials,
complete elliptic integrals, and other special functions.
The output can be used directly for interval computations
or to generate provably correct floating-point approximations in any format.
Performance is competitive with earlier arbitrary-precision software, and
sometimes orders of magnitude faster.
We also partially cover the generalized hypergeometric function~${}_pF_q$
and computation of high-order parameter derivatives.
\end{abstract}

%\begin{keywords}
%hypergeometric functions, interval arithmetic, arbitrary-precision arithmetic, Bessel functions, orthogonal polynomials, automatic differentiation
%\end{keywords}

%\begin{AMS}
% 33F05, % Special functions: Numerical approximation and evaluation
% 33C20, % Generalized hypergeometric series, pFq
% 33C05, % Classical hypergeometric functions 2F1
% 33C10, % Bessel and Airy, cylinder functions, 0F1
% 33C15, % Confluent hypergeometric functions, Whittaker functions, 1F1
% 65G30, % Interval and finite arithmetic
% 65Y20, % Complexity and performance of numerical algorithms
% 65D20, % Computation of special functions, construction of tables
% 97N80, % Mathematical software, computer programs
% 33B15, % Gamma, beta and polygamma functions
% 33B20, % Incomplete beta and gamma functions (error functions, probability integral, Fresnel integrals)
% 33C45 % Orthogonal polynomials
% %33C55, % Spherical harmonics
% %33C75  % Elliptic integrals as hypergeometric functions
%\end{AMS}

\section{Introduction}
Naive numerical methods for evaluating special functions are prone
to give inaccurate results, particularly in the complex domain.
Ensuring just a few correct digits in double precision (53-bit) floating-point arithmetic
can be hard, and many scientific applications require even higher accuracy~\cite{bailey2015high}.
To give just a few examples, Bessel functions are needed with
quad precision (113-bit) accuracy in some
electromagnetics and hydrogeophysics simulations~\cite{6529388,Kuhlman2015},
with hundreds of digits in integer relation searches
such as for
Ising class integrals~\cite{bailey2006integrals},
and with thousands of digits in number theory
for computing with L-functions and modular forms~\cite{Booker2006}.

In this work, we address evaluation of
the confluent hypergeometric functions ${}_0F_1$, ${}_1F_1$ and~${}_2F_0$
(equivalently, the Kummer $U$-function)
and the Gauss hypergeometric function ${}_2F_1$ for complex parameters
and argument, to any accuracy and with rigorous error bounds.
Based on these very general functions, we are in turn able to compute
incomplete gamma and beta functions, Bessel functions, Legendre functions, and others,
covering a large portion of the special functions in standard references
such as Abramowitz and Stegun~\cite{AbramowitzStegun1964}
and the Digital Library of Mathematical Functions~\cite{NIST:DLMF}.

The implementation is part of the C library
Arb~\cite{Johansson2013arb}\footnote{\url{https://github.com/fredrik-johansson/arb/}},
which is free software (GNU LGPL).
The code is extensively tested and documented,
thread-safe, and runs on common 32-bit and 64-bit platforms.
Arb is a standard package in SageMath~\cite{sage},
which provides a partial high-level interface.
Partial Python and Julia bindings are also available.
Interfacing is easy from many other languages, including Fortran/C++.

\section{Hypergeometric functions}

A function $f(z) = \sum_{k=0}^{\infty} c(k) z^k$ is called hypergeometric
if the Taylor coefficients $c(k)$
form a \emph{hypergeometric sequence}, meaning that they satisfy a first-order
recurrence relation
$c(k+1) = R(k) c(k)$ where the term ratio $R(k)$
is a rational function of $k$.

The product (or quotient)
of two hypergeometric sequences with respective term ratios $R_1(k), R_2(k)$
is hypergeometric with term ratio $R_1(k) R_2(k)$ (or $R_1(k) / R_2(k)$).
Conversely, by factoring $R(k)$,
we can write any hypergeometric
sequence in closed form using powers $z^{k+1} = z \cdot z^k$
and gamma functions $\Gamma(a+k+1) = (a+k) \Gamma(a+k)$,
times a constant determined by the initial value of the recurrence.
The rising factorial
$\quad (a)_k = a (a+1) \cdots (a+k-1) = \Gamma(a+k) / \Gamma(a)$
is often used instead of the gamma function, depending on the initial value.
A standard notation for hypergeometric functions is offered by
the \emph{generalized hypergeometric function} (of order $(p, q)$)
\begin{equation}
{}_pF_q(a_1,\ldots,a_p, b_1,\ldots,b_q, z) = \sum_{k=0}^{\infty} T(k), \quad
T(k) = \frac{(a_1)_k \cdots (a_p)_k}{(b_1)_k \cdots (b_q)_k} \frac{z^k}{k!}
\label{eq:pfq}
\end{equation}
or the
\emph{regularized generalized hypergeometric function}
\begin{equation}
{}_p{\tilde F}_q(a_1,\ldots,a_p, b_1,\ldots,b_q, z) = 
\sum_{k=0}^{\infty} \frac{(a_1)_k \cdots (a_p)_k}{\Gamma(b_1+k) \cdots \Gamma(b_q+k)} \frac{z^k}{k!}
= \frac{{}_pF_q(\ldots)}{\Gamma(b_1) \cdots \Gamma(b_q)}
\label{eq:regpfq}
\end{equation}
where $a_i$ and $b_i$ are called (upper and lower) parameters, and
$z$ is called the argument
(see \cite[Chapter~16]{NIST:DLMF} and \cite{wolf}). Both \eqref{eq:pfq} and \eqref{eq:regpfq}
are solutions of the linear ODE
\begin{equation}
\left[z\prod_{n=1}^{p}\left(z\frac{d}{dz} + a_n\right) - z\frac{d}{dz}\prod_{n=1}^{q}\left(z\frac{d}{dz} + b_n-1\right) \right] f(z) = 0.
\label{eq:pfqode}
\end{equation}

\subsection{Analytic interpretation}

Some clarification is needed to interpret the formal
sums \eqref{eq:pfq} and \eqref{eq:regpfq} as analytic functions.
If any $a_i \in \mathbb{Z}_{\le 0}$,
the series terminates as ${}_pF_q = \sum_{k=0}^{-a_i} T(k)$, a
polynomial in~$z$ and a rational function of the other parameters.
If any $b_j \in \mathbb{Z}_{\le 0}$, ${}_pF_q$ is generally undefined
due to dividing by zero,
unless some $a_i \in \mathbb{Z}_{\le 0}$ with $a_i > b_j$ in which case
it is conventional to use the truncated series.

If some $a_i = b_j \in \mathbb{Z}_{\le 0}$ and the series does
not terminate earlier, \eqref{eq:pfq} is ambiguous. One
possible interpretation
is that we can cancel $a_i$ against~$b_j$ to get a
generalized hypergeometric function of order $(p-1, q-1)$.
Another interpretation is that $0 / 0 = 0$, terminating the series.
Some implementations are inconsistent and may use either interpretation.
For example, Mathematica evaluates
${}_1F_1(-n,-n,z) = e^z$ and ${}_1F_1(-1,-1,z) = 1+z$.
We do not need this case, and leave it undefined.
Ambiguity can be avoided by working with ${}_p{\tilde F}_q$, which
is well-defined for all values of the lower parameters,
and with explicitly truncated hypergeometric series
when needed.

With generic values of the parameters,
the rate of convergence of the series~\eqref{eq:pfq} or \eqref{eq:regpfq} depends on the sign of $p - q + 1$,
giving three distinct cases.

Case $p \le q$: the series converges for any $z$
and defines an entire function
with an irregular (exponential) singularity at $z = \infty$. The
trivial example is ${}_0F_0(z) = \exp(z)$.
The \emph{confluent hypergeometric
functions} ${}_0F_1(z)$ and ${}_1F_1(z)$
form exponential integrals, incomplete gamma functions, Bessel functions,
and related functions.

Case $p = q + 1$: the series converges for $|z| < 1$.
The function
is analytically continued
to the (cut) complex plane, with regular (algebraic or logarithmic)
singularities at $z = 1$ and $z = \infty$.
The \emph{principal branch} is defined
by placing a branch cut on the real interval $(1,+\infty)$. We define
the function value on the branch cut itself by continuity coming
from the lower half plane,
and the value at $z = 1$ as the limit from the left, if it exists.
The \emph{Gauss
hypergeometric function} ${}_2F_1$ is the most important example, forming
various orthogonal polynomials and integrals of algebraic functions.
The principal branch is chosen to be consistent with elementary evaluations such as
${}_2F_1(a,1;1;z) = {}_1F_0(z) = (1-z)^{-a} = \exp(-a \log(1-z))$
and $z \, {}_2F_1(1,1; 2; z) = -\log(1-z)$
with the usual convention that $\operatorname{Im}(\log(z)) \in (-\pi,\pi]$.

Case $p > q + 1$: the series only converges at $z = 0$, but
can be viewed as an asymptotic expansion
valid when $|z| \to 0$. 
Using resummation theory (Borel summation), it can be 
associated to an analytic function of $z$.

\subsection{Method of computation}

By \eqref{eq:pfqode}, hypergeometric functions are D-finite (holonomic), i.e.\ satisfy linear
ODEs with rational function coefficients.
There is a theory of ``effective analytic continuation''
for general D-finite functions~\cite{david1990computer,van2001fast,vdH:hol,Mezzarobba2010,Mezzarobba2011,MezzarobbaSalvy2010}.
Around any point $z \in \mathbb{C} \cup \{ \infty \}$,
one can construct a basis of solutions of the ODE
consisting of generalized formal power series
whose terms satisfy a linear recurrence relation
with rational function coefficients (the function ${}_pF_q$ arises in the special case where
the recurrence at $z = 0$ is hypergeometric,
that is, has order 1).
The expansions permit numerical evaluation of local solutions.
A D-finite function defined by an ODE and initial values
at an arbitrary point $z_0$
can be evaluated at any point $z_k$ by connecting local solutions
along a path $z_0 \to z_1 \to z_2 \ldots \to z_k$.

This is essentially the Taylor method for integrating ODEs numerically, but
D-finite functions are special.
First, the special form of the series expansions permits
rapid evaluation using reduced-complexity algorithms.
%Second, error bounds can be computed algorithmically.
Second, by use of generalized series expansions with singular prefactors, $z_0$ and $z_k$ can
be singular points (including~$\infty$),
or arbitrarily close to singular points, without essential loss of efficiency.

The general algorithm for D-finite functions is quite complicated and has
never been implemented fully with rigorous error bounds, even restricted
to hypergeometric functions. The most difficult problem is
to deal with irregular singular points, where the local series expansions
become asymptotic (divergent), and resummation theory is needed.
Even at regular points, it is difficult to perform all steps efficiently.
The state of the art is Mezzarobba's package~\cite{Mezzarobba2016},
which covers regular singular points.

In this work, we implement ${}_pF_q$ and ${}_p{\tilde F}_q$
rigorously using the direct series expansion at $z = 0$.
This is effective when $p \le q$ as long as $|z|$ is not too large,
and when $p = q + 1$ as long as $|z| \ll 1$.
Further, and importantly,
we provide a complete implementation of the cases $p \le 2, q \le 1$,
permitting efficient evaluation for any~$z$.
The significance of order $p \le 2, q \le 1$ is that many
second-order ODEs that arise in applications can be transformed to this case of~\eqref{eq:pfqode}.
We are able to cover evaluation
for any~$z$ due to the fact that the expansions of these ${}_pF_q$'s at $z = \infty$
(and $z = 1$ for~${}_2F_1$) can be expressed as finite linear combinations of other
${}_pF_q$'s with $p \le 2, q \le 1$
via explicit connection formulas. This includes the Borel regularized function
${}_2F_0$, which is related to the Kummer $U$-function.
Evaluation of ${}_2F_0$ near $z = 0$ (which is used for the
asymptotic expansions of ${}_0F_1$ and ${}_1F_1$ at $z = \infty$)
is possible thanks to explicit error bounds already available in the literature.
Analytic continuation via the hypergeometric ODE is used in one special case when computing ${}_2F_1$.

Using hypergeometric series as the main building block allows us to
cover a range of special functions efficiently
with reasonable effort.
For ${}_pF_q$'s of higher order, this simplified
approach is no longer possible. With some exceptions, expansions
of ${}_pF_q$ at $z \ne 0$ are not hypergeometric, and methods to
compute rigorous error bounds for asymptotic series
have not yet been developed into concrete algorithms.
Completing the picture for the higher ${}_pF_q$ functions
should be a goal for future research.

\subsection{Parameter derivatives and limits}

Differentiating ${}_pF_q$ with respect to~$z$ simply shifts parameters,
and batches of high-order $z$-derivatives are easy to compute using recurrence relations.
In general, derivatives with respect to parameters have no convenient closed forms.
We implement parameter derivatives
using truncated power series arithmetic (automatic differentiation). In other words,
to differentiate $f(a)$ up to order $n-1$, we compute $f(a + x)$ using
arithmetic in the ring $\mathbb{C}[[x]] / \langle x^{n} \rangle$.
This is generally more efficient than numerical differentiation, particularly for large $n$.
Since formulas involving analytic functions on $\mathbb{C}$ translate directly to
$\mathbb{C}[[x]]$, we can avoid
symbolically differentiated formulas, which often become unwieldy.

The most important use for parameter derivatives is to compute limits
with respect to parameters. Many connection formulas have removable
singularities at some parameter values. For example, if $f(a,z) = g(a,z) / \sin(\pi a)$ and
$g(a,z) = 0$ when $a \in \mathbb{Z}$, we
compute $\lim_{\varepsilon \to 0} f(a+\varepsilon)$ when $a \in \mathbb{Z}$
by evaluating
$g(a+\varepsilon) / \sin(\pi (a + \varepsilon))$
in $\mathbb{C}[[\varepsilon]] / \langle \varepsilon^2 \rangle$,
formally cancelling the zero constant terms in the power series division.

\subsection{Previous work}

Most published work on numerical methods for special functions uses heuristic error estimates.
Usually, only a subset of a function's domain is covered correctly.
Especially if only machine precision is used,
expanding this set is a hard problem that requires a patchwork of methods,
e.g.\ integral representations,
uniform asymptotic expansions, continued fractions, and recurrence relations.
A~good survey of methods for ${}_1F_1$
and ${}_2F_1$ has been done by Pearson et al.~\cite{pearson2009computation,pearson2014numerical}.
The general literature is too vast to summarize here; see~\cite{NIST:DLMF} for a bibliography.

Rigorous implementations until now have only supported a small set of special
functions on a restricted domain.
The arbitrary-precision libraries
MPFR, MPC and MPFI~\cite{Fousse2007,enge2011mpc,RevolRouillier2005}
provide elementary functions
and a few higher transcendental functions
of real variables, with guaranteed correct rounding. Other rigorous implementations
of restricted cases
include~\cite{du,du2002hypergeometric,yamamoto2005,colman2011validated}.
Computer algebra systems and arbitrary-precision libraries
such as  Mathematica, Maple, Maxima, Pari/GP, mpmath
and MPFUN~\cite{wolfmma,maplesoft,mpmath,PARI2,maxima,bailey2015mpfun2015}
support a wide range of special functions for complex variables,
but all use heuristic error estimates and sometimes produce incorrect output.
Performance can also be far from satisfactory.

Our contribution is to simultaneously support (1) a wide range of special functions,
(2) complex variables, (3) arbitrary precision, and (4) rigorous error bounds, with (5) high speed.
To achieve these goals,
we use interval arithmetic to automate most error bound
calculations, pay attention to asymptotics, and implement previously-overlooked optimizations.

The point of reference for speed is mainly other
arbitrary-precision software, since
the use of software arithmetic with variable precision and unlimited exponents
adds perhaps a factor 100 baseline
overhead compared to hardware floating-point arithmetic.
The goal is first of all to maintain reasonable speed even when
very high precision is required or when function arguments lie in numerically
difficult regions.
With further work, the overhead at low precision could also be reduced.

\section{Arbitrary-precision interval arithmetic}

Interval arithmetic provides a rigorous way to compute with real numbers~\cite{tucker2011validated}.
Error bounds are propagated automatically through the whole computation,
completely accounting for rounding errors as well as the effect of uncertainty in initial values.

Arb uses the midpoint-radius form of interval arithmetic (``ball arithmetic'')
with an arbitrary-precision floating-point midpoint and a low-precision floating-point
radius, as in
$\pi \in [3.14159265358979323846264338328 \pm 1.07 \cdot 10^{-30}]$.
This allows tracking error bounds without
significant overhead compared to floating-point arithmetic~\cite{vdH:ball}.
Of course, binary rather than decimal numbers are used internally.
We represent complex numbers in rectangular form as pairs of real balls;
in some situations, it would be better to use true complex balls (disks) with a single radius.

The drawback of interval arithmetic
is that error bounds must be overestimated.
Output intervals can be correct but
\emph{useless}, e.g.\ $\pi \in [0 \pm 10^{123}]$.
The combination of variable precision and interval arithmetic
allows increasing the precision until the output is useful~\cite{RevolRouillier2005}.
As in plain floating-point arithmetic, it helps to use algorithms
that are numerically stable, giving tighter enclosures,
but this is mainly a matter of efficiency (allowing lower precision to be used) and not of correctness.

\subsection{Adaptivity}

When the user asks for $P$ bits of precision,
Arb chooses internal evaluation parameters
(such as working precision and series cutoffs) to attempt to achieve $2^{-P}$
relative error, but stops
after doing $O(\operatorname{poly}(P))$ work, always returning a correct
interval within a predictable amount of time, where the output
may be useless if convergence is too slow,
cancellations too catastrophic, the input intervals too imprecise, etc.
A wrapper program or the end user will typically treat the
interval implementation as a black box and try
with, say, $P_1 = 1.1 P$ bits of precision, followed by precisions
$P_k = 2P_{k-1}$ if necessary until the output error bound has converged to
the desired tolerance of $2^{-P}$.
It is easy to abort and signal an error to the end user or
perhaps try a different implementation if
this fails after a reasonable number of steps.
For example, evaluating the incomplete gamma
function $\Gamma(20i, 10\pi i)$ at 32, 64 and 128 bits with Arb gives:
\begin{equation*}
\begin{array}{c}
{}[\pm 5.26 \cdot 10^{-13}] + [\pm 5.28 \cdot 10^{-13}]i \\
{}[4.0 \cdot 10^{-17} \pm 2.88 \cdot 10^{-19}] + [-1.855 \cdot 10^{-15} \pm 5.57 \cdot 10^{-19}]i \\
{}[4.01593625943 \cdot 10^{-17} \pm 2.58 \cdot 10^{-29}] + [-1.8554278420570 \cdot 10^{-15} \pm 6.04 \cdot 10^{-29}]i
\end{array}
\end{equation*}

Increasing the precision is only effective if the user can
provide exact input or intervals of width about $2^{-P}$.
We do not address the harder problem of bounding a function tightly
on a ``wide'' interval such as $[\pi,2\pi]$. Subdivision works,
but the worst-case cost for a resolution of $2^{-P}$ increases exponentially with $P$.
For some common special functions
($J_0(z), \operatorname{Ai}(z), \operatorname{erf}(z)$, etc.),
one solution is to evaluate the
function at the interval midpoint or endpoints and use
monotonicity or derivative bounds to enclose
the range on the interval.
We have used such methods for a few particular functions,
with good results, but do not treat the problem in general here.

We attempt to ensure heuristically that the output radius
converges to 0 when $P \to \infty$, but this is not formally guaranteed,
and some exceptions exist.
For example, when parameters $a \in \mathbb{Z}$ and $a \not \in \mathbb{Z}$
are handled separately, convergence might fail with input like $a = [3 \pm 2^{-P}]$.
Such limitations could be removed with further effort.

\subsection{Floating-point output}

It takes only a few lines of wrapper
code around the interval version of a function
to get floating-point output in any format, with prescribed guaranteed accuracy.
A C header file is available
to compute hypergeometric and other special functions
accurately for the C99 \texttt{double complex} type.%\footnote{\url{https://github.com/fredrik-johansson/arbcmath}}.

The correct rounding of the last bit
can normally be deduced by computing to sufficiently high precision.
However, if the true value of the function is an exactly representable floating-point
number (e.g.\ 3.25 = 3.250000\ldots) and the algorithm used to compute
it does not generate a zero-width interval, this
iteration fails to terminate (``the table-maker's dilemma'').
Exact points are easily handled in trivial cases
(for example, $z = 0$ is the only case for $e^z$).
However, hypergeometric functions can have nontrivial exact
points, and detecting them in general is a hard problem.

\subsection{Testing}

Every function has a test program that generates many (usually $10^3$ to $10^6$)
random inputs. Inputs
are distributed non-uniformly, mixing real and imaginary parts that are
exact, inexact, tiny, huge, integer, near-integer, etc., to trigger
corner cases with high probability. For each input, a function value
is computed in two different ways, by varying the precision, explicitly
choosing different internal algorithms, or applying a recurrence
relation to shift the input.
The two computed intervals, say $I_1, I_2$, must then satisfy
$I_1 \cap I_2 \ne \emptyset$, which is checked.
In our experience, this form of testing is very powerful,
as one can see by inserting deliberate bugs to verify
that the test code fails.
Other forms of testing
are also used.

\section{Direct evaluation of hypergeometric series}

\label{sect:hypdirsum}

The direct evaluation of ${}_pF_q$ via \eqref{eq:pfq}
involves three tasks: selecting the number of terms $N$,
computing the truncated sum $S(N) = \sum_{k=0}^{N-1} T(k)$, and
(unless the series terminates at this point) bounding the error
${}_pF_q - S(N)$, which is equal to $\sum_{k=N}^{\infty} T(k)$ when
the series converges.

In Arb, $N$ is first selected heuristically with
53-bit hardware arithmetic (with some care to avoid overflow and other issues).
In effect, for $P$-bit precision,
linear search is used to pick the first $N$ such that
$|T(N)| < \max_{n < N} |T(n)| / 2^P$.
If no such $N$ exists up to a precision-dependent limit $N \le N_{\text{max}}$,
the $N$ that minimizes $|T(N)|$ subject to $N \le N_{\text{max}}$ is chosen
($N_{\text{max}}$
allows us to compute a crude bounding interval for
${}_pF_q$ instead of getting stuck if the series converges too slowly).

Both $S(N)$ and $T(N)$ are subsequently computed
using interval arithmetic, and the value of $T(N)$ is used to compute a
rigorous bound for the tail.

\subsection{Tail bounds for convergent series}

\label{sect:tails}

If $N$ is so large that $|T(k+1)/T(k)|$ is small for all $k \ge N$,
then $T(N) (1 + \varepsilon)$ is a good
estimate of the error.
It is not hard to turn this observation into an effective bound.
Here, we define $b_{q+1} = 1$ so that
$T(k) = z^k \prod_{i=1}^p (a_i)_k / \prod_{i=1}^{q+1} (b_i)_k$
without the separate factorial.

\begin{theorem}
\label{thm:pfqbound}
With $T(k)$ as in \eqref{eq:pfq},
if $p \le q + 1$ and $\operatorname{Re}(b_i+N) > 0$ for all $b_i$,
then $\left|\sum_{k=N}^{\infty} T(N)\right| \le C |T(N)|$ where
\begin{equation*}
C = \begin{cases} \frac{1}{1-D} & D < 1 \\ \infty & D \ge 1, \end{cases}, \quad
D = |z| \, \prod_{i=1}^p \left(1 + \frac{|a_i-b_i|}{|b_i+N|}\right) \prod_{i=p+1}^{q+1} \frac{1}{|b_i + N|}.
\end{equation*}
\end{theorem}

\begin{proof}
Looking at
the ratio $T(k+1) / T(k)$ for $k \ge N$, we cancel out upper and lower parameters
$\frac{|a+k|}{|b+k|} = \left| 1 + \frac{a-b}{b+k} \right| \le 1 + \frac{|a-b|}{|b+N|}$
and bound remaining lower parameter factors as
$\left|\frac{1}{b+k}\right| \le \frac{1}{|b+N|}.$
Bounding the tail by a geometric series gives the result.
\end{proof}

The same principle is used to get tail bounds for
parameter derivatives of \eqref{eq:pfq},
i.e.\ bounds for each coefficient in
$\sum_{k=N}^{\infty} T(k) \in \mathbb{C}[[x]] / \langle x^n \rangle$
given $a_i, b_i, z \in \mathbb{C}[[x]] / \langle x^n \rangle$.
First, we fix some notation: if $A \in \mathbb{C}[[x]]$,
$A_{[k]}$ is the coefficient of $x^k$,
$A_{[m:n]}$ is the power series $\sum_{k=m}^{n-1} A_{[k]} x^k$,
$|A|$ denotes $\sum_{k=0}^{\infty} |A_{[k]}| x^k$ which can be viewed as an element of $\mathbb{R}_{\ge 0}[[x]]$,
and $A \le B$ signifies that $|A|_{[k]} \le |B|_{[k]}$ holds for all $k$.
Using
$(A B)_{[k]} = \sum_{j=0}^k A_{[j]} B_{[k-j]}$ and $(1 / B)_{[k]} = (1 / B_{[0]}) \sum_{j=1}^k -B_{[j]} (1/B)_{[k-j]}$, it is easy to check that
$|A + B| \le |A| + |B|$,
$|A B|  \le |A| |B|$ and
$|A / B| \le |A| / \mathcal{R}(B)$
where $\mathcal{R}(B) = |B_{[0]}| - |B_{[1:\infty]}|$.
\cref{thm:pfqbound} can now be restated for power series:
if $p \le q + 1$ and $\operatorname{Re}({b_i}_{[0]}\,+\,N) > 0$ for all $b_i$,
then $\left|\sum_{k=N}^{\infty} T(N)\right| \le C |T(N)|$ where
\begin{equation*}
C = \begin{cases} \frac{1}{\mathcal{R}(1-D)} & D_{[0]} < 1 \\ \infty & D_{[0]} \ge 1, \end{cases}, \quad
D = |z| \, \prod_{i=1}^p \left(1 + \frac{|a_i-b_i|}{\mathcal{R}(b_i+N)}\right) \prod_{i=p+1}^{q+1} \frac{1}{\mathcal{R}(b_i + N)}.
\end{equation*}

To bound sums and products of power series with (complex) interval coefficients,
we can use floating-point upper bounds with directed rounding for the absolute values instead
of performing interval arithmetic throughout.
For $\mathcal{R}(B)$, we must pick a lower bound for $|B_{[0]}|$ and upper bounds for
the coefficients of $|B_{[1:\infty]}|$.

\subsection{Summation algorithms}

Repeated use of the forward recurrence
$S(k+1) = T(k) S(k)$, $T(k+1) = R(k) T(k)$
where $R(k) = T(k+1) / T(k)$ with initial values
$S(0) = 0, T(0) = 1$ yields $S(N)$ and $T(N)$. This requires
$O(N)$ arithmetic
operations in $\mathbb{C}$, where we consider $p, q$ fixed,
or $\widetilde{O}(N^2)$ bit operations in the common situation where $N \sim P$,
$P$ being the precision.
When $N$ is large, Arb uses three different series evaluation algorithms
to reduce the complexity, depending on the output precision
and the bit lengths of the inputs.
Here we only give a brief overview of the methods; see~\cite{brent1976complexity,Borwein1987,ChudnovskyChudnovsky1988,Smith1989,Haible1998,vdH:hol,Ziegler2005,BostanGaudrySchost2007,Bernstein2008,mca,Johansson2014rectangular,Johansson2014thesis}
for background and theoretical analysis.

\emph{Binary splitting} (BS) is used at high precision when all parameters $a_i, b_i$
as well as the
argument $z$ have short binary representations, e.g. $a = 1+i, z = 3.25 = 13 \cdot 2^{-2}$. The
idea is to compute the matrix product $M(N-1) \cdots M(0)$ where
\begin{equation*}\begin{pmatrix} T(k+1) \\ S(k+1)\end{pmatrix} =
\begin{pmatrix} R(k) & 0 \\ 1 & 1 \end{pmatrix} \begin{pmatrix} T(k) \\ S(k)\end{pmatrix} \equiv M(k) \begin{pmatrix} T(k) \\ S(k)\end{pmatrix}
\end{equation*}
using a divide-and-conquer strategy,
and clearing denominators so that only a single final division is necessary.
BS reduces the asymptotic complexity of computing $S(N)$ exactly
(or to $\widetilde{O}(N)$ bits)
to $\widetilde{O}(N)$ bit operations.
BS is also used at low precision when $N$ is large,
since the $O(\log N)$ depth of operand dependencies
often makes it more numerically stable than the forward recurrence.

\emph{Rectangular splitting} (RS) is used at high precision when all parameters
have short binary representations but the argument has a long binary representation,
e.g.\ $z = \pi$ (approximated to high precision). The idea
is to write $\sum_{k=0}^{N-1} c_k z^k = (c_0 + c_1 z + \ldots + c_{m-1} z^{m-1})
+ z^m (c_m + c_{m+1} z + \ldots + c_{2m-1} z^{m-1}) + z^{2m} \ldots$ where $m \sim \sqrt{N}$,
reducing the number of multiplications where both factors depend
on $z$ to $O(\sqrt{N})$.
One extracts the common factors from the
coefficients $c_k$ so that all other multiplications and divisions only
involve short coefficients.
Strictly speaking, RS does not reduce the theoretical complexity
except perhaps by a factor $\log N$,
but the speedup in practice can grow to more than a factor 100,
and there is usually some speedup even at modest precision and with small $N$.
Another form of RS should be used when a single parameter
has a long binary representation~\cite{Johansson2014rectangular}; this
is a planned future addition that has not yet been implemented in Arb for general hypergeometric series,

\emph{Fast multipoint evaluation} (FME) is used at high precision when not
all parameters have short binary representations. Asymptotically,
it reduces the complexity to $\widetilde{O}(\sqrt{N})$ arithmetic operations.
Taking $m \sim \sqrt{N}$,
one computes $M(m+k-1) \cdots M(k)$ as a matrix with entries in $\mathbb{C}[k]$,
evaluates this matrix at $k = 0, m, 2m, \ldots, m(m-1)$,
and multiplies the evaluated matrices together.
FME relies on asymptotically fast FFT-based polynomial arithmetic, which
is available in Arb. Unlike BS and RS,
its high overhead and poor numerical stability~\cite{KohlerZiegler2008}
limits its use to very high precision.

Arb attempts to choose the best algorithm automatically.
Roughly speaking, BS and RS are used at precision above 256 bits, provided that
the argument and parameters are sufficiently short,
and FME is used above 1500 bits otherwise.
Due to the many variables involved, the automatic choice will not always
be optimal.

We note that Arb uses an optimized low-level implementation of RS
to compute elementary functions at precisions up to a few thousand bits~\cite{Johansson2015elementary}.
Although the elementary function series expansions are hypergeometric, that
code is separate from the implementation of generic hypergeometric
series discussed in this paper. Such optimizations could be
implemented for other ${}_pF_q$ instances.

\subsubsection{Parameter derivatives}

A benefit of using power series arithmetic instead of explicit formulas
for parameter derivatives is that complexity-reduction techniques
immediately apply.
Arb implements both BS and RS over $\mathbb{C}[[x]] / \langle x^n \rangle$,
with similar precision-based cutoffs as over $\mathbb{C}$.
FME is not yet implemented.

RS is currently only used for $n \le 2$, due to a minor technical limitation.
Binary splitting is always used when $n > 20$, as it allows taking advantage
of asymptotically fast polynomial multiplication; in particular, the
complexity is reduced to $\widetilde{O}(n)$ arithmetic operations in $\mathbb{C}$
when all inputs have low degree as polynomials in $x$.

\subsubsection{Regularization}

When computing ${}_p{\widetilde F}_q$, the recurrence relation
is the same, but the initial value
changes to $T(0) = (\Gamma(b_1) \cdots \Gamma(b_q))^{-1}$.
Above, we have implicitly assumed that no
$b_i \in \mathbb{Z}_{\le 0}$, so that
we never divide by zero in $R(k)$.
If some $-b_i \in \mathbb{Z}_{\ge 0}$ is part of the summation range
(assuming that the series does not terminate on $N < -b_i$),
the recurrence must be started at $T(k)$ at $k = 1 - b_i$ to avoid dividing by zero.
With arithmetic in $\mathbb{C}$, all the terms $T(0), \ldots, T(1-b_i)$ are then zero.
With arithmetic in $\mathbb{C}[[x]] / \langle x^n \rangle$, $n > 1$,
the use of the recurrence must be restarted whenever the coefficient of $x^0$ in
$b_i + k$ becomes zero. 
With $n > 1$, the terms $T(k)$ between such points need 
not be identically zero, so the recurrence may have to be started and
stopped several times.
For nonexact intervals that intersect $\mathbb{Z}_{\le 0}$,
the problematic terms should be skipped the same way.
Arb handles such cases, but there is room for optimization in the implementation.

When computing ${}_0{\tilde F}_1(c,z)$, ${}_1{\tilde F}_1(a,c,z)$ or ${}_2{\tilde F}_1(a,b,c,z)$
directly over $\mathbb{C}$ and $c \in \mathbb{Z}_{\le 0}$,
working in $\mathbb{C}[[x]] / \langle x^2 \rangle$ is avoided by using a direct formula, e.g.\
\begin{equation*}
{}_2{\tilde F}_1(a,b,-n,z) = \frac{(a)_{n+1} (b)_{n+1} z^{n+1}}{(n+1)!} \, {}_2F_1(a+n+1, b+n+1, n+2, z).
\end{equation*}

\section{The gamma function}

Computation of $\Gamma(s)$ is crucial for hypergeometric functions.
Arb contains optimized versions
of each of the functions $\Gamma(s)$, 
$1 / \Gamma(s)$ (avoiding division by zero when $s \in \mathbb{Z}_{\le 0}$),
$\log \Gamma(s)$ (with the correct branch structure), and
$\psi(s) = \Gamma'(s) / \Gamma(s)$,
for each of the domains
$s \in \mathbb{R}, \mathbb{C}, \mathbb{R}[[x]], \mathbb{C}[[x]]$.

Small integer and half-integer $s$ are handled directly.
A separate function is provided for $\Gamma(s)$
with $s \in \mathbb{Q}$, with optimizations
for denominators $3, 4, 6$, including use of
elliptic integral identities \cite{BorweinZucker1992}.
Euler's constant $\gamma = -\psi(1)$
is computed using the Brent-McMillan algorithm,
for which rigorous error bounds were
derived in~\cite{BrentJohansson2013bound}.

Many algorithms for $\Gamma(s), s \in \mathbb{C}$ have been proposed,
including \cite{Lanczos1964,Borwein1987,Spouge1994,schmelzer2007computing},
but in our experience, it is hard to beat the asymptotic Stirling series
\begin{equation}
\label{eq:stirling}
\log \Gamma(s) = \left(s-\frac{1}{2}\right)\log(s) - s +
      \frac{\log(2 \pi)}{2}
        + \sum_{k=1}^{N-1}  \frac{B_{2k}}{2k(2k-1)s^{2k-1}}
      + R(N,s)
\end{equation}
with argument reduction based on $\Gamma(s+r) = (s)_r \Gamma(s)$
and $\Gamma(s)\Gamma(1-s) = \pi / \sin(\pi s)$,
where the error term~$R(N,s)$ is bounded as in~\cite{Olver1997}.
Conveniently, \eqref{eq:stirling} is numerically
stable for large~$|s|$ and gives the correct branch structure for $\log \Gamma(s)$.

Fast evaluation of the rising factorial $(s)_r = s (s+1) \cdots (s+r-1)$
is important for the argument reduction when $|s|$ is small.
Arb uses the RS algorithm
described in~\cite{Johansson2014rectangular},
which builds on previous work by Smith~\cite{Smith2001}.
It also uses BS when~$s$ has a short binary representation,
or when computing derivatives of high order~\cite{Johansson2014thesis}.

The drawback of the Stirling series is that Bernoulli numbers are required.
An alternative for moderate $|s|$ is to use the approximation by a lower
gamma function $\Gamma(s) \approx \gamma(s,N) = s^{-1} N^a e^{-N} {}_1F_1(1,1+s,N)$
with a large $N$.
The methods for fast evaluation of hypergeometric series apply.
However, as noted in~\cite{Johansson2014rectangular}, this is usually
slower than the Stirling series if Bernoulli numbers are cached
for repeated evaluations.

\subsection{Bernoulli numbers and zeta constants}

Arb caches Bernoulli numbers, but generating them the first time
could be time-consuming if not optimized.
The best method in practice uses
$B_{2n} = (-1)^{n+1} 2(2n)! \zeta(2n) (2\pi)^{-2n}$
together with the von Staudt-Clausen theorem
to recover exact numerators.
Instead of using the Euler product for $\zeta(2n)$ as in~\cite{fillebrown1992faster},
the L-series is used directly.
The observation made in~\cite{bloemen}
is that if one computes a batch of Bernoulli numbers
in descending order $B_{2n}, B_{2n-2}, B_{2n-4}, \ldots$, the powers in
$\zeta(2n) = \sum_{k=1}^{\infty} k^{-2n}$
can be recycled, i.e.\ $k^{-2n} = k^2 \cdot k^{-(2n+2)}$.
On an Intel i5-4300U, Arb generates all the exact Bernoulli
numbers up to $B_{10^3}$ in 0.005 s,
$B_{10^4}$ in 1 s and $B_{10^5}$ in 10 min;
this is far better than
we have managed with other
methods, including methods with lower asymptotic complexity.

For $z \in \mathbb{Z}$, Arb computes $\Gamma(z+x) \in \mathbb{R}[[x]]$ via the Stirling series
when $|z| > P / 2$ and via
$\Gamma(1-x) = \exp\left(\gamma x + \sum_{k=2}^\infty \zeta(k) x^{k} / k \right)$
when $|z| \le P / 2$ (at $P$-bit precision).
The $\zeta(2n+1)$-constants are computed using the Euler product for large $n$,
and in general using the convergence acceleration method of~\cite{Borwein2000},
with BS at high precision when $n$ is small and with a term recurrence as in~\cite{mpfralg}
otherwise. For multi-evaluation, term-recycling is used \cite[Algorithm~4.7.1]{Johansson2014thesis};
much like in the case of Bernoulli numbers, this
has lower overhead than any other known method, including the algorithms proposed in~\cite{BorweinBradleyCrandall2000}.
The Stirling series with BS over $\mathbb{R}[[x]]$ could also be used
for multi-evaluation of $\zeta(2n+1)$-constants,
but this would only be competitive when computing thousands of constants
to over $10^5$ bits. As noted in \cite{Johansson2014thesis}, the Stirling series is better
for this purpose than
the $\Gamma(x) \approx \gamma(x,N)$ approximation
used in~\cite{Karatsuba1998,BorweinBradleyCrandall2000}.

We mention $\zeta(n)$-constants since they are of independent interest,
e.g.\ for convergence acceleration of series~\cite{flajolet1996zeta}.
In fact, such methods apply to slowly converging hypergeometric series~\cite{bogolubsky2006fast,skorokhodov2005method},
though we do not pursue this approach here.

\section{Confluent hypergeometric functions}

\label{sect:hyp1f1}

The conventional basis of solutions of
$z f''(z) + (b-z) f'(z) - a f(z) = 0$
consists of the confluent hypergeometric functions (or Kummer functions)
${}_1F_1(a,b,z)$ and $U(a,b,z)$, where $U(a,b,z) \sim z^{-a}, |z| \to \infty$.

\Cref{tab:specfunconfluent} gives a list of functions implemented in Arb
that are expressible via ${}_1F_1$ and $U$.
In this section, we outline the evaluation approach.
Other functions that could be implemented in the same way
include the Whittaker functions, parabolic cylinder functions,
Coulomb wave functions, spherical Bessel functions, Kelvin functions,
and the Dawson and Faddeeva functions. Indeed, the user can easily
compute these functions via the functions already available in~\cref{tab:specfunconfluent},
with the interval arithmetic automatically taking care of error bounds.

\begin{table}
\begin{center}
\begin{tabular}{l l}
Function & Notation \\ \hline
Confluent hypergeometric function \rule{0pt}{0.4cm} & ${}_1F_1{}(a;b;z), {}_1{\tilde F}_1{}(a;b;z)$, $U(a,b,z)$ \\
Confluent hypergeometric limit function & ${}_0F_1{}(b;z), {}_0{\tilde F}_1{}(b;z)$ \\
Bessel functions & $J_{\nu}(z)$, $Y_{\nu}(z)$, $I_{\nu}(z)$, $K_{\nu}(z)$ \\
Airy functions & $\operatorname{Ai}(z), \operatorname{Ai}'(z)$, $\operatorname{Bi}(z), \operatorname{Bi}'(z)$ \\
Error function & $\operatorname{erf}(z)$, $\operatorname{erfc}(z)$, $\operatorname{erfi}(z)$ \\
Fresnel integrals & $S(z)$, $C(z)$ \\
Incomplete gamma functions & $\Gamma(s,z), \gamma(s,z), P(s,z), Q(s,z), \gamma^{*}(s,z)$ \\
Generalized exponential integral & $E_{\nu}(z)$ \\
Exponential and logarithmic integral & $\operatorname{Ei}(z)$, $\operatorname{li}(z)$, $\operatorname{Li}(z)$ \\
Trigonometric integrals & $\operatorname{Si}(z)$, $\operatorname{Ci}(z)$, $\operatorname{Shi}(z)$, $\operatorname{Chi}(z)$ \\
Laguerre function (Laguerre polynomial) & $L^{\mu}_{\nu}(z)$ \\
Hermite function (Hermite polynomial) & $H_{\nu}(z)$ \\
\end{tabular}%}
\caption{Implemented variants and derived cases of confluent hypergeometric functions.}
\label{tab:specfunconfluent}
\end{center}
\end{table}

\subsection{Asymptotic expansion}

It turns out to be convenient to define the function
$U^{*}(a,b,z) = z^a U(a,b,z)$, which is asymptotic to~1 when $|z| \to \infty$.
We have the important asymptotic expansion
\begin{equation}
\label{eq:uasymp2f0}
U^{*}(a,b,z) = {}_2F_0\left(a, a-b+1, -\frac{1}{z}\right) = \sum_{k=0}^{N-1} \frac{(a)_k (a-b+1)_k}{k! (-z)^k} + \varepsilon_N(a,b,z),
\end{equation}
where $|\varepsilon_N(a,b,z)| \to 0$ for fixed $N,a,b$ when $|z| \to \infty$.
The ${}_2F_0$ series is divergent when $N \to \infty$ (unless
$a$ or $a-b+1 \in \mathbb{Z}_{\le 0}$ so that it reduces to a polynomial),
but it is natural to \emph{define} the function ${}_2F_0(a,b,z)$
for all $a,b,z \in \mathbb{C}$ in terms of $U$ via \eqref{eq:uasymp2f0}.
The choice between ${}_2F_0$ and $U$ (or $U^{*}$) is then just a matter of notation.
It is well known that this definition is equivalent
to Borel resummation of the ${}_2F_0$ series.

We use Olver's bound for the error term $|\varepsilon_N|$, implementing
the formulas almost verbatim from \cite[13.7]{NIST:DLMF}.
We do not reproduce the formulas here since lengthy case distinctions are needed.
As with convergent hypergeometric series, the error is bounded by $|T(N)|$ times
an easily computed function.
To choose $N$, the same heuristic is used as with convergent series.

We have not yet implemented parameter derivatives of the asymptotic series,
since the main use of parameter derivatives
is for computing limits in formulas involving non-asymptotic
series. Parameter derivatives of $\varepsilon_N$
could be bounded using the Cauchy integral formula.

\subsection{Connection formulas}

For all complex $a, b$ and all $z \ne 0$,
\begin{equation}
\label{eq:mconnect}
U(a,b,z)
    = \frac{\Gamma(1-b)}{\Gamma(a-b+1)} {}_1F_1(a,b,z)
    + \frac{\Gamma(b-1)}{\Gamma(a)} z^{1-b} {}_1F_1(a-b+1,2-b,z)
\end{equation}
\begin{equation}
\label{eq:uconnect}
{}_1{\tilde F}_1(a,b,z)
    = \frac{(-z)^{-a}}{\Gamma(b-a)} U^{*}(a,b,z)
    + \frac{z^{a-b} e^z}{\Gamma(a)} U^{*}(b-a,b,-z)
\end{equation}
with the understanding that principal branches are used everywhere and
$U(a,b,z) = \lim_{\varepsilon \to 0} U(a,b+\varepsilon,z)$
in \eqref{eq:mconnect} when $b \in \mathbb{Z}$.
Formula \eqref{eq:mconnect}, which allows
computing $U$ when $|z|$ is too small to use the asymptotic series,
is \cite[13.2.42]{NIST:DLMF}.
Formula \eqref{eq:uconnect}, which in effect gives us the asymptotic
expansion for ${}_1F_1$, can be derived from
the connection formula between ${}_1F_1$ and $U$
given in \cite[13.2.41]{NIST:DLMF}.
The advantage of using $U^{*}$
and the form~\eqref{eq:uconnect} instead of $U$ is that it
behaves continuously in interval arithmetic when $z$ straddles the real axis.
The discrete jumps that occur in~\eqref{eq:uconnect} when crossing
branch cuts on the right-hand side
only contribute by an exponentially small amount:
when $|z|$ is large enough for the asymptotic expansion to be used,
$z^{a-b}$ is continuous where $e^z$ dominates and
$(-z)^{-a}$ is continuous where $e^z$ is negligible.

\subsection{Algorithm selection}

Let $P$ be the target precision in bits. To support unrestricted
evaluation of ${}_1F_1$ or $U$,
the convergent series should be used for fixed $z$ when $P \to \infty$,
and the asymptotic series should be used for fixed $P$ when $|z| \to \infty$.
Assuming that $|a|,|b| \ll P, |z|$, the
error term in the asymptotic series is smallest approximately
when $N = |z|$, and the magnitude of the error
then is approximately $e^{-N} = e^{-|z|}$.
In other words, the asymptotic series can give
up to $|z| / \log(2)$ bits of accuracy.
Accordingly, to compute either ${}_1F_1$ or $U$, there are four main cases:

\begin{itemize}
\item For ${}_1F_1$ when $|z| < P \log(2)$, use the ${}_1F_1$ series directly.
\item For ${}_1F_1$ when $|z| > P \log(2)$, use \eqref{eq:uconnect} and evaluate two ${}_2F_0$ series (with error bounds for each).
\item For $U$ or $U^{*}$ when $|z| > P \log(2)$, use the ${}_2F_0$ series directly.
\item For $U$ or $U^{*}$ when $|z| < P \log(2)$, use \eqref{eq:mconnect} and evaluate two ${}_1F_1$ series.
If $b \in \mathbb{Z}$, compute $\lim_{\varepsilon \to 0} U(a,b+\varepsilon,z)$
by substituting $b \to b+\varepsilon \in \mathbb{C}[[\varepsilon]] / \langle \varepsilon^2 \rangle$
and evaluating the right hand side of~\eqref{eq:mconnect} using power series arithmetic.
\end{itemize}

The cutoff based on $|z|$ and $P$ is correct asymptotically,
but somewhat naive since it
ignores cancellation in the ${}_1F_1$
series and in the connection formula \eqref{eq:mconnect}.
In the transition region $|z| \sim P$, all significant bits may be lost.
The algorithm selection can be optimized by estimating the amount of
cancellation with either formula, and selecting the one that should
give the highest final accuracy.
For example, assuming that both $a,b$ are small, ${}_1F_1(a,b,z) \approx {}_1F_1(1,1,z) = e^z$
while the terms in the ${}_1F_1$ series grow to about $e^{|z|}$,
so $(|z| - \operatorname{Re}(z)) / \log(2)$ bits are lost to cancellation,
while no cancellation occurs in the ${}_2F_0$ series.
A more advanced scheme should take into account $a$ and~$b$.
The algorithm selection in Arb generally uses the simplified
asymptotic estimate, although cancellation estimation has been implemented for
a few specific cases. This is an important place for future optimization.

In general, the Kummer transformation ${}_1F_1(a,b,z) = e^z {}_1F_1(b-a,b,-z)$
is used to compute ${}_1F_1$ for $\operatorname{Re}(z) < 0$, so
that worst-case cancellation occurs around
the oscillatory region $z = iy$, $y \in \mathbb{R}$.
An interesting potential optimization would be to use methods
such as \cite{chevillard2013multiple} to reduce cancellation there also.

\subsection{Bessel functions}

Bessel functions are computed via the ${}_0F_1$ series for small $|z|$
and via $U^{*}$ for large $|z|$.
The formula ${}_0F_1(b,z^2) = e^{-2z} {}_1F_1(b-\tfrac{1}{2},2b-1,4z)$
together with \eqref{eq:uconnect} gives the asymptotic expansions of
\begin{equation*}
J_{\nu}(z), I_{\nu}(z) = \frac{1}{\Gamma(\nu+1)} \left(\frac{z}{2}\right)^{\nu}
             {}_0F_1\left(\nu+1, \mp \frac{z^2}{4}\right).
\end{equation*}
To compute ${}_0F_1(b,z)$ itself when $|z|$ is large, $J$ and $I$ are used
according to the sign of midpoint of $\operatorname{Re}(z)$,
to avoid evaluating square roots on the branch cut. For $K_{\nu}(z)$,
we use
$K_{\nu}(z) = \left(2z/\pi\right)^{-1/2} e^{-z}
    U^{*}(\nu+\tfrac{1}{2}, 2\nu+1, 2z)$
when $|z|$ is large, and
\begin{equation}
\label{eq:besselk}
K_{\nu}(z) = \frac{1}{2} \frac{\pi}{\sin(\pi \nu)} \left[
            \left(\frac{z}{2}\right)^{-\nu}
                {}_0{\widetilde F}_1\left(1-\nu, \frac{z^2}{4}\right)
             -
             \left(\frac{z}{2}\right)^{\nu}
                 {}_0{\widetilde F}_1\left(1+\nu, \frac{z^2}{4}\right)
            \right]
\end{equation}
otherwise, with a limit computation when $\nu \in \mathbb{Z}$.
Note that it would be a mistake to use \eqref{eq:besselk} for all $z$,
not only because it requires evaluating four asymptotic series instead of one,
but because it would lead to exponentially large cancellation when $z \to +\infty$.

\subsection{Other functions}

The other functions in \cref{tab:specfunconfluent} are generally computed via
${}_1F_1$, ${}_1F_2$, ${}_2F_2$ or ${}_2F_3$ series for small $|z|$
(in all cases, ${}_1F_1$ could be used, but the alternative series
are better),
and via $U$ or $U^{*}$ for large $|z|$.
Airy functions, which are related to Bessel functions with parameter
$\nu = \pm 1/3$, use a separate implementation that does not rely
on the generic code for ${}_0F_1$ and $U^{*}$.
This is done as an optimization since the generic code for hypergeometric series
currently does not have an interface to input rational parameters with non-power-of-two denominators
exactly.

In all cases, care is taken to use formulas that avoid cancellation asymptotically when $|z| \to \infty$,
but cancellation can occur in the transition region $|z| \sim P$.
Currently, the functions
$\operatorname{erf}, \operatorname{erfc}, S, C, \operatorname{Ai}, \operatorname{Ai}', \operatorname{Bi}, \operatorname{Bi}'$
automatically evaluate the function at the midpoint of the input interval
and compute a near-optimal error bound based on derivatives,
automatically increasing the internal working precision to compensate
for cancellation. This could also be done for other functions in the future,
most importantly for exponential integrals and Bessel functions of small order.

\section{The Gauss hypergeometric function}

\label{sect:hyp2f1}

\begin{table}
\begin{center}
\begin{tabular}{l l}
Function & Notation \\ \hline
Hypergeometric function \rule{0pt}{0.4cm} & ${}_2F_1{}(a,b;c;z), {}_2{\tilde F}_1{}(a,b;c;z)$ \\
Chebyshev functions & $T_{\nu}(z), U_{\nu}(z)$ \\
Jacobi function & $P_{\nu}^{\alpha,\beta}(z)$ \\
Gegenbauer (ultraspherical) function & $C_{\nu}^{\mu}(z)$ \\
Legendre functions & $P_{\nu}^{\mu}(z), Q_{\nu}^{\mu}(z), \mathcal{P}_{\nu}^{\mu}(z), \mathcal{Q}_{\nu}^{\mu}(z)$ \\
Spherical harmonics & $Y_n^m(\theta,\varphi)$ \\
Incomplete beta function & $B(a,b;z), I(a,b;z)$ \\
Complete elliptic integrals & $K(z), E(z)$ \\
\end{tabular}%}
\caption{Implemented variants and derived cases of the Gauss hypergeometric function.}
\label{tab:specfungauss}
\end{center}
\end{table}

The function~${}_2F_1$ is implemented for general parameters, together
with various special cases shown in \cref{tab:specfungauss}.
For Legendre functions, two different branch cut variants are implemented.
We list $T_{\nu}(z), U_{\nu}(z), K(z), E(z)$
for completeness; Chebyshev functions are generally computed
using trigonometric formulas (or binary exponentiation
when $\nu \in \mathbb{Z}$), and complete elliptic integrals
are computed using arithmetic-geometric mean iteration.
Some other functions in \cref{tab:specfungauss}
also use direct recurrences to improve speed or
numerical stability in special cases, with the ${}_2F_1$ representation
used as a general fallback.

\subsection{Connection formulas}

The function ${}_2F_1$ can be rewritten
using the Euler and Pfaff transformations
\begin{align}
\label{eq:hyp2f1transf1}
\begin{split}
{}_2F_1 (a,b,c,z) &= (1-z)^{c-a-b} {}_2F_1 (c-a, c-b, c, z) \\
                  &= (1-z)^{-a} \,{}_2F_1 \!\left (a, c-b, c, \frac{z}{z-1} \right) \\
                  &= (1-z)^{-b} \,{}_2F_1 \!\left (c-a, b, c, \frac{z}{z-1} \right).
\end{split}
\end{align}
It can also be written as a linear combination of two ${}_2F_1$'s
of argument $1/z$, $1/(1-z)$, $1-z$ or $1-1/z$ (see \cite[15.8]{NIST:DLMF}).
For example, the $1/z$ transformation reads
\begin{align}
\label{eq:2f1inv}
\begin{split}
\frac{\sin(\pi(b-a))}{\pi} {}_2{\tilde F}_1(a,b,c,z) & =
\frac{(-z)^{-a}}{\Gamma(b) \Gamma(c-a)} \,{}_2{\tilde F}_1\!\left(a,a-c+1,a-b+1,\frac{1}{z}\right) \\
& -
\frac{(-z)^{-b}}{\Gamma(a) \Gamma(c-b)} \,{}_2{\tilde F}_1\!\left(b,b-c+1,b-a+1,\frac{1}{z}\right).
\end{split}
\end{align}
The first step when computing ${}_2F_1$ is to check whether
the original series is terminating, or whether one of \eqref{eq:hyp2f1transf1}
results in a terminating series.
Otherwise, we pick the linear fractional transformation among
\begin{equation*}
z, \quad \frac{z}{z-1}, \quad \frac{1}{z}, \quad \frac{1}{1-z}, \quad 1-z, \quad 1-\frac{1}{z}
\end{equation*}
that results in the argument of smallest absolute value,
and thus the most rapid convergence of the hypergeometric series.
In Arb, experimentally-determined tuning factors are used in this
selection step to account for the fact that
the first two transformations require half as many function evaluations.
The coverage of the complex plane is illustrated in \cref{fig:domains2f1}.
This strategy is effective for all complex $z$ except
near $e^{\pm \pi i / 3}$, which we handle below.

\begin{figure}
\centering
\includegraphics[scale=0.3]{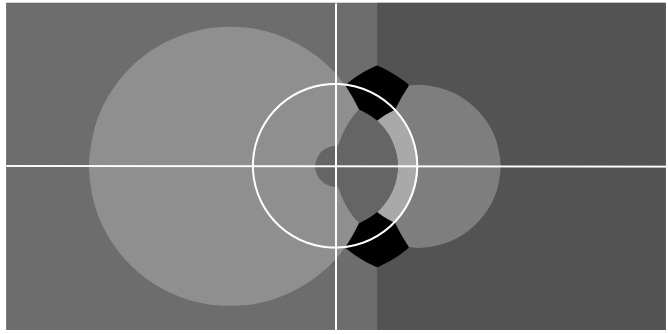}
\caption{Different shades correspond to different fractional transformations chosen
depending on the location of $z$ in the complex plane. The coordinate axes
and the unit circle are drawn in white. Black indicates the corner cases where
no transformation is effective.}
\label{fig:domains2f1}
\end{figure}

\subsection{Parameter limits}

The $1/z$ and $1/(1-z)$ transformations involve a division by
$\sin(\pi(b-a))$, and the $1-z$ and $1-1/z$ transformations
involve a division by $\sin(\pi(c-a-b))$.
Therefore, when $b-a$ or $c-a-b$ respectively is an integer, we
compute $\lim_{\varepsilon \to 0} {}_2F_1(a + \varepsilon,b,c,z)$
using
$\mathbb{C}[[\varepsilon]] / \langle \varepsilon^2 \rangle$ arithmetic.

The limit computation can only be done automatically when $b-a$ (or $c-a-b$)
is an exact integer.
If, for example, $a, b$ are intervals representing $\pi, \pi+1$ and
$b-a$ thus is a nonexact interval containing~1, the output would not be finite.
To solve this problem, Arb allows the user to pass
boolean flags as input to the ${}_2F_1$
function, indicating that $b-a$ or $c-a-b$ represent exact integers
despite $a,b,c$ being inexact.

\subsection{Corner case}

When $z$ is near $e^{\pm \pi i / 3}$, we use analytic continuation.
The function $f(z) = {}_2F_1(a,b,c,z_0+z)$ satisfies
\begin{equation}
%P_2(z) f''(z) + P_1 f'(z) + P_0 f(z) = 0
f''(z) = -\frac{P_1(z)}{P_2(z)} f'(z) - \frac{P_0(z)}{P_2(z)} f(z).
\label{eq:fdiffeq}
\end{equation}
with
$P_2(z) = (z+z_0)(z+z_0-1)$,
$P_1(z) = (a+b+1)(z+z_0) - c$,
$P_0(z) = ab$.
%\begin{align*}
%P_2(z) &= (z+z_0)(z+z_0-1) \\
%P_1(z) &= (a+b+1)(z+z_0) - c \\
%P_0(z) &= ab.
%\end{align*}
It follows from \eqref{eq:fdiffeq} that the coefficients
in the Taylor series $f(z) = \sum_{k=0}^{\infty} f_{[k]} z^k$ satisfy
the second-order recurrence equation
\begin{equation}
\label{eq:rrec}
R_2(k) f_{[k+2]} + R_1(k) f_{[k+1]} + R_0(k) f_{[k]} = 0
\end{equation}
with
$R_2(k) = (k+1)(k+2)(z_0-1)z_0$,
$R_1(k) = (k+1) (2k+a+b+1) z_0 - (k+1)(k+c)$,
$R_0(k) = (a+k)(b+k)$.

Knowing $f(0), f'(0)$, we can compute $f(z), f'(z)$ if $z$ is sufficiently small
and $z_0 \not\in \{0,1\}$,
thanks to \eqref{eq:rrec} and the truncation bound given below.
We use the usual ${}_2F_1$ series at the origin, following by
two analytic continuation steps
\begin{equation*}
0 \quad \to \quad 0.375 \pm 0.625i \quad \to \quad 0.5 \pm 0.8125i \quad \to \quad z.
\end{equation*}
This was found to be a good choice based on experiments.
There is a tradeoff involved in choosing the evaluation points,
since smaller steps yield faster local convergence.

Currently, BS and other reduced-complexity
methods are not implemented for this Taylor series in Arb.
At high precision, assuming that the parameters $a,b,c$
have short representations,
it would be better to choose a finer path in which
$z_k$ approaches the final evaluation point as $2^{-k}$,
with BS to evaluate the Taylor series (this is known
as the bit-burst method).

Though only used by default in the corner case, the user can
invoke analytic continuation along an arbitrary path; for example,
by crossing $(1,\infty)$, a non-principal branch can be computed.

A simple bound for the coefficients in the Taylor series is obtained
using the Cauchy-Kovalevskaya majorant method,
following~\cite{van2001fast}.

\begin{theorem}
If $f(z) = \sum_{k=0}^{\infty} f_{[k]} z^k$ satisfies \eqref{eq:fdiffeq} and $z_0 \not \in \{0,1\}$, then
for all $k \ge 0$,
$|f_{[k]}| \le A \binom{N+k}{k} \nu^k$
provided that
\begin{equation}
\label{eq:2f1boundnu} \nu \ge \max\left(\frac{1}{|z_0-1|}, \frac{1}{|z_0|}\right), \quad M_0 \ge 2 \nu |ab|, \quad M_1 \ge \nu (|a+b+1| + 2|c|),
\end{equation}
\begin{align}
\label{eq:2f1boundN} N &\ge \frac{\max(\sqrt{2 M_0}, 2 M_1)}{\nu}, \\
\label{eq:2f1boundA} A &\ge \max\left(|f_{[0]}|, \frac{|f_{[1]}|}{\nu (N+1)}\right).
\end{align}
\end{theorem}

\begin{proof}
Note that
\begin{equation*}
\frac{1}{P_2(z)} = \frac{1}{z+z_0-1} - \frac{1}{z+z_0} =
\sum_{k=0}^{\infty} \left(\frac{1}{(z_0-1)^{k+1}} - \frac{1}{z_0^{k+1}}\right) (-z)^k.
\end{equation*}
Accordingly, with $M_0, M_1, \nu$ as in \eqref{eq:2f1boundnu},
\begin{equation*}
\left|\left(-\frac{P_i(z)}{P_2(z)}\right)_{[k]}\right| \le M_i \nu^k = \left(\frac{M_i}{1-\nu z}\right)_{[k]},
\quad i = 0,1.
\end{equation*}
In other words, \eqref{eq:fdiffeq} is majorized by
$g''(z) = \frac{M_1}{1-\nu z} g'(z) + \frac{M_0}{1-\nu z} g(z).$
Using \eqref{eq:2f1boundN} in turn gives the majorizing differential equation
\begin{equation}
h''(z) = \frac{N+1}{2} \left( \frac{\nu}{1-\nu z} \right) h'(z)
+ \frac{(N+1)N}{2} \left( \frac{\nu}{1-\nu z} \right)^2 h(z).
\label{eq:gdiffeq}
\end{equation}
The simple solution $h(z) = A (1-\nu z)^{-N}$
of \eqref{eq:gdiffeq} now provides the bound
$|f_{[k]}| \le h_{[k]} = A \binom{N+k}{k} \nu^k$
with $A$ chosen as in \eqref{eq:2f1boundA} for the initial values.
\end{proof}

Note that the bound
is a hypergeometric sequence, indeed $h(z) = A \, {}_1F_0(N,\nu z)$,
so the tail of the Taylor series
is easily bounded as in \cref{sect:tails}.

A different method to evaluate ${}_2F_1$ in the corner case
is used in Mathematica, Maxima, mpmath and perhaps
others~\cite{maxgosp,vogtgosp}. No error bounds have been published
for that method; we encourage further investigation.

\subsection{Numerical stability}

Unlike the confluent hypergeometric functions, cancellation is not a problem
for any $z$ as long as the parameters are small. With small complex parameters
and 64-bit precision, about 5-10 bits are lost on most of
the domain, up to about 25 bits in the black regions of \cref{fig:domains2f1},
so even machine precision interval arithmetic would be adequate for many
applications. However, cancellation does become a problem with large
parameters (e.g.\ for orthogonal polynomials of high degree).
The current implementation could be improved for some cases by
using argument transformations and recurrence relations to minimize
cancellation.

\subsection{Higher orders}

The functions ${}_3F_2$, ${}_4F_3$, etc.\
have a $1/z$ transformation analogous to \eqref{eq:2f1inv}, but no other
such formulas for generic parameters~\cite[16.8]{NIST:DLMF}.
We could cover $|z| \gg 1$ by evaluating ${}_pF_q$-series directly,
but other methods are needed on the annulus surrounding the unit circle.
Convergence accelerations schemes such
as~\cite{willis2012acceleration,bogolubsky2006fast,skorokhodov2005method}
can be effective, but the D-finite analytic continuation approach
(with the singular expansion at $z = 1$) is likely
better since effective error bounds are known and since
complexity-reduction techniques apply. A study remains to be done.

We note that an important special function related to ${}_{n+1}F_n$,
the polylogarithm $\operatorname{Li}_s(z)$, is implemented for general complex $s, z$
in Arb, using direct series
for $|z| \ll 1$ and the Hurwitz zeta function otherwise~\cite{Johansson2014hurwitz}.
Also $s$-derivatives are supported.

\section{Benchmarks}

We compare Arb (git version as of June 2016) to software for specific functions,
and to Mathematica 10.4 which supports arbitrary-precision evaluation of
generalized hypergeometric functions.
Tests were run on an Intel i5-4300U CPU
except Mathematica which was run on an Intel i7-2600 (about 20\% faster).
All timing measurements have been rounded to two digits.\footnote{Code for some of the benchmarks: \url{https://github.com/fredrik-johansson/hypgeom-paper/}}

\subsection{Note on Mathematica}

We compare to Mathematica since it generally appears
faster than Maple and mpmath.
It also attempts to track numerical errors,
using a form of significance arithmetic~\cite{Sofroniou2005precise}.
With \texttt{f=Hypergeometric1F1}, the commands \texttt{f[-1000,1,N[1,30]]},
\texttt{N[f[-1000,1,1]]}, and \texttt{N[f[-1000,1,1],30]}
give $0.155$, $0.154769$, and $0.154769339118406535633854462041$,
where all results show loss of significance being tracked correctly.
Note that \texttt{N[..., d]} attempts to produce~$d$ significant digits
by adaptively using a higher internal precision.
For comparison, at 64 and 128 bits, Arb produces
$[\pm 5.51 \cdot 10^6]$ and $[0.154769339118 \pm 9.35 \cdot 10^{-13}]$.

Unfortunately, Mathematica's heuristic error tracking is unreliable.
For example, \texttt{f[-1000,1,1.0]}, produces $-1.86254761018 \cdot 10^9$
without any hint that the result is wrong.
The input \texttt{1.0} designates a machine-precision
number, which in Mathematica is distinct from an arbitrary-precision number,
potentially disabling error tracking in favor of speed (one of many pitfalls that the user must
be aware of). However, even
arbitrary-precision results obtained from exact input cannot be trusted:

\begin{small}
\begin{verbatim}
In[1]:= a=8106479329266893*2^-53; b=1/2 + 3900231685776981*2^-52 I;
In[2]:= f=Hypergeometric2F1[1,a,2,b];
In[3]:= Print[N[f,16]]; Print[N[f,20]]; Print[N[f,30]];                                                                   
0.9326292893367381 + 0.4752101514062250 I
0.93263356923938324111 + 0.47520053858749520341 I
0.932633569241997940484079782656 + 0.475200538581622492469565624396 I
In[4]:= f=Log[-Re[Hypergeometric2F1[500I,-500I,-500+10(-500)I,3/4]]];
In[5]:= $MaxExtraPrecision=10^4; Print[{N[f,100],N[f,200],N[f,300]}//N];
{2697.18, 2697.18, 2003.57}
\end{verbatim}
\end{small}

As we see, even comparing answers computed at two levels
of precision (100 and 200 digits) for consistency is not reliable.
Like with most numerical software (this is not a critique of Mathematica
in particular), Mathematica users
must rely on external knowledge to divine whether results are likely to be correct.

\subsection{Double precision}

Pearson gives a list of 40 inputs $(a,b,z)$ for ${}_1F_1$
and 30 inputs $(a,b,c,z)$ for ${}_2F_1$, chosen to exercise different
evaluation regimes in IEEE 754 double precision
implementations~\cite{pearson2009computation,pearson2014numerical}.
We also test $U(a,b,z)$ with the ${}_1F_1$ inputs
and the Legendre function $Q_a^c(1-2z)$
with the ${}_2F_1$ inputs (with argument $1-2z$, it is equivalent to a linear
combination of two ${}_2F_1$'s of argument $z$).

\begin{table}
\renewcommand{\arraystretch}{1.4}
\setlength{\tabcolsep}{.4em}
\begin{center}
\begin{scriptsize}
\begin{tabular}{l l | r r | l}
          & Code & Average & Median   & Accuracy \\
\hline
${}_1F_1$ & SciPy         & 2.7 & 0.76         & 18 good, 4 fair, 4 poor, 5 wrong, 2 NaN, 7 skipped\\
${}_2F_1$ & SciPy         & 24   & 0.56        & 18 good, 1 fair, 1 poor, 3 wrong, 1 NaN, 6 skipped \\
\hline
${}_2F_1$ & Michel \& S.  & 7.7  & 2.1 & 22 good, 1 poor, 6 wrong, 1 NaN \\
\hline
${}_1F_1$ & MMA (m)       & 1100      &     29 & 34 good, 2 poor, 4 wrong, 2 no significant digits out \\
${}_2F_1$ & MMA (m)       & 30000     &     72   & 29 good, 1 fair \\

$U$       & MMA (m)       & 4400      & 190   &  28 good, 4 fair, 2 wrong, 6 no significant digits out \\
$Q$       & MMA (m)       & 4300      &   61  &  21 good, 3 fair, 2 poor, 1 wrong, 3 NaN \\
\hline
${}_1F_1$ & MMA (a)       & 2100      & 170   & 39 good, 1 not good as claimed (actual error $2^{-40}$) \\
${}_2F_1$ & MMA (a)       & 37000     & 540   & 30 good ($2^{-53}$) \\
$U$       & MMA (a)       & 25000     & 340   & 38 good, 2 not as claimed ($2^{-40}, 2^{-45}$) \\
$Q$       & MMA (a)       & 8300      & 780   & 28 good, 1 not as claimed ($2^{-25}$), 1 wrong \\
\hline
${}_1F_1$ & Arb           &  200      & 32    & 40 good (correct rounding) \\
${}_2F_1$ & Arb           &  930      & 160   & 30 good (correct rounding) \\
$U$       & Arb           & 2000      & 93    & 40 good (correct rounding) \\
$Q$       & Arb           & 3000      & 210   & 30 good ($2^{-53}$) \\
\end{tabular}
\end{scriptsize}
\caption{Time (in microseconds) per function evaluation, averaged over the 30 or 40 inputs, and resulting accuracy.
MMA is Mathematica, with (m) machine precision and (a) arbitrary-precision arithmetic.
With SciPy, Michel \& S. and MMA (m), we deem a result with relative error at most $2^{-40}$ good,
$2^{-20}$ fair, $2^{-1}$ poor, otherwise wrong. With MMA (a), good is $<2^{-50}$.}
\label{tab:pearson2}
\end{center}
\end{table}

With Arb, we measure the time to compute certified correctly rounded 53-bit
floating-point values, except for the function $Q$ where we compute
the values to a certified relative error of $2^{-53}$ before
rounding (in the current version,
some exact outputs for this function are not recognized,
so the correct-rounding loop would not terminate).
We interpret all inputs as double precision constants rather than
real or complex numbers that would need to be enclosed by intervals.
For example, $0.1 = 3602879701896397 \cdot 2^{-55}$, $e^{i \pi/3} = 2^{-1} + 3900231685776981 \cdot 2^{-52} i$.
This has no real impact on this particular benchmark, but it is
simpler and demonstrates Arb as a black box to implement floating-point functions.

For ${}_2F_1$, we compare with the double precision C++ implementation
by Michel and Stoitsov \cite{michel2008fast}.
We also compare with the Fortran-based ${}_1F_1$ and ${}_2F_1$ implementations in
SciPy~\cite{scipy}, which only support inputs with real parameters.

We test Mathematica in two ways: with machine precision numbers as input,
and using its arbitrary-precision arithmetic to attempt to get 16 significant digits.
\texttt{N[]} does not work well
on this benchmark, because Mathematica gets stuck expanding huge symbolic
expressions when some of the inputs are given exactly (\texttt{N[]}
also evaluates some of those resulting expressions incorrectly),
so we implemented a precision-increasing loop like the one used
with Arb instead.

Since some test cases take much longer than others,
we report both the average and the median time for a function
evaluation with each implementation.

The libraries using pure machine arithmetic are fast, but
output wrong values in several cases.
Mathematica with machine-precision input is only slightly more reliable.
With arbitrary-precision numbers, Mathematica's significance arithmetic
computes most values correctly, but the output is not always correct
to the 16 digits Mathematica claims. One case,
$Q_{500}^{500}(11/5)$, is completely wrong.
Mathematica's ${}_2F_1$ passes,
but as noted earlier, would fail with the built-in \texttt{N[]}.
Arb is comparable to Mathematica's machine precision in
median speed (around $10^4$ function evaluations per second),
and has better worst-case speed,
with none of the wrong results.

\subsection{Complex Bessel functions}

Kodama~\cite{Kodama2011} has implemented the functions $J_{\nu}(z), Y_{\nu}(z)$ and
$H^{(1),(2)}_{\nu}(z) = J_{\nu} \pm i Y_{\nu}(z)$ for complex $\nu, z$
accurately (but without formal error bounds) in Fortran.
Single, double and extended (targeting $\approx 70$-bit accuracy) precision are supported.
Kodama's self-test program evaluates all four functions
with parameter $\nu$ and $\nu + 1$ for 2401 pairs $\nu,z$ with real
and imaginary parts between $\pm 60$, making
19\,208 total function calls. Timings are shown in \cref{tab:kodama}.

\begin{table}
\renewcommand{\arraystretch}{1.1}
\setlength{\tabcolsep}{.4em}
\begin{center}
\begin{small}
%\begin{tabular}{c c c c c}
%Bits   & Kodama     & Mathematica   & Arb    & Arb (4-in-1)  \\ \hline
%53     & 4.7      & 271         & 10   & 5.1               \\
%113    & 275      & 343         & 11   & 5.9               \\
%\end{tabular}
%\begin{small}
\begin{tabular}{l | r r | l}
Code               & Time  & 4-in-1 &  Accuracy  \\ \hline
Kodama, 53-bit     & 4.7  &   &  19208 good ($\approx 2^{-45}$) according to self-test          \\
Kodama, $\approx$ 80-bit   & 270  &   &  19208 good ($\approx 2^{-68}$) according to self-test          \\ \hline
MMA (machine)            & 75  &   &  14519 good, 607 fair, 256 poor, 3826 wrong \\
MMA (\texttt{N[]}, 53-bit)           & 270   &  &  17484 good, 2 fair,   273 poor, 1449 wrong \\
MMA (\texttt{N[]}, 113-bit)           & 340  &   &  18128 good, 206 fair, 187 poor, 687 wrong \\ \hline
Arb (53-bit)       & 10  & 5.1  &  19208 good (correct rounding) \\
Arb (113-bit)      & 11  & 5.9  &  19208 good (correct rounding) \\
\end{tabular}
\end{small}
\caption{Time (in seconds) to compute the 19\,208 Bessel function values in Kodama's test.
The good/fair/poor/wrong thresholds are $2^{-40}$ (53-bit) or $2^{-100}$ (113-bit) / $2^{-20}$ / $2^{-1}$.}
\label{tab:kodama}
\end{center}
\end{table}

We compiled Kodama's code with GNU Fortran 4.8.4, which uses a quad precision
(113-bit or 34-digit) type for extended precision.
We use Arb to compute the functions on the same
inputs, to double and quad precision
with certified correct rounding of both the real and imaginary parts.
In the 4-in-1 column for Arb, we time
computing the four values $J, Y, H^{(1)}, H^{(2)}$ simultaneously rather than with
separate calls, still with correct rounding for each function.
Surprisingly, our implementation is competitive with Kodama's double
precision code,
despite ensuring correct rounding.
There is a very small penalty going from double to quad precision,
due to the fact that a working precision of 200-400 bits is used
in the first place for many of the inputs.

We also test Mathematica in three ways: with machine precision input,
and with exact input using \texttt{N[...,16]} or \texttt{N[...,34]}.
Mathematica computes many values incorrectly. For example,
with $\nu=-53.9-53.4i, z = -54.7+17.61i$, the three evaluations
of the \texttt{HankelH1} function give three different results
\begin{equation*}
\begin{array}{c}
H^{(1)}_{\nu}(z) = 1.30261 \cdot 10^{-34} - 4.49948 \cdot 10^{-35} i, \\
H^{(1)}_{\nu}(z) = -1.18418492459404 \cdot 10^{-32} + 2.805568990224272 \cdot 10^{-31} i, \\
H^{(1)}_{\nu}(z) = -3.893447525697409211107221229269630 \cdot 10^{-25} + \\
6.133044639987209608932345865755910 \cdot 10^{-25} i
\end{array}
\end{equation*}
while the correct value is about $1.65 \cdot 10^{-27} + 2.28 \cdot 10^{-27} i$.

\subsection{High precision}

\Cref{tab:prectimings} shows the time to compute Bessel functions
for small $\nu$ and $z$,
varying the precision and the type of the inputs.
In A, both $\nu$ and $z$ have few bits, and Arb uses BS.
In B, $\nu$ has few bits but $z$ has full precision, and RS is used.
In C, both $\nu$ and $z$ have full precision, and FME is used.
Since much of the time in C is spent computing $\Gamma(\pi+1)$,
D tests $J_{\pi}(\pi)$ without this factor.
E and F involve computing the parameter derivatives
of two ${}_0F_1$ functions with $\mathbb{C}[[x]] / \langle x^2 \rangle$
arithmetic to produce $K_3(z)$, respectively using BS and RS.

In all cases, complexity-reducing methods give a notable speedup at high precision.
Only real values are tested; complex numbers (when of similar
``difficulty'' on the function's domain)
usually increase running times uniformly by a factor 2--4.

\begin{table}
\renewcommand{\arraystretch}{1.4}
\setlength{\tabcolsep}{.4em}
\begin{center}
\begin{scriptsize}
\begin{tabular}{l | r r r r r | r r r r r}
   &  \multicolumn{5}{|c|}{Mathematica}  &  \multicolumn{5}{|c}{Arb}    \\
  & 10 & 100 & 1000 & 10\,000 & 100\,000 & 10 & 100 & 1000 & 10\,000 & 100\,000 \\ \hline
A: $J_3(3.25)$ &
24 & 48 & 320 & 9200 & 590\,000 &
6 & 27 & 140 & 1600 & 26\,000 \\
B: $J_3(\pi)$ &
23 & 56 & 800 & 110\,000 & 17\,000\,000 &
7 & 28 & 320 & 14\,000 & 960\,000 \\
C: $J_{\pi}(\pi)$ &
58 & 220 & 34\,000 & 8\,500\,000 &  &
12 & 91 & 2800 & 270\,000 & 44\,000\,000 \\
D: ${}_0F_1(\pi\!+\!1,-\tfrac{\pi^2}{4})\!$ &
26 & 84 & 1800 & 350\,000 & 61\,000\,000 &
7 & 49 & 1500 & 93\,000 & 14\,000\,000 \\
E: $K_3(3.25)$ &
68 & 160 & 1700 & 140\,000 & 19\,000\,000 &
40 & 150 & 1300 & 20\,000 & 500\,000 \\
F: $K_3(\pi)$ &
69 & 160 & 2600 & 350\,000 & 52\,000\,000 &
43 & 170 & 1900 & 67\,000 & 4\,100\,000 \\
\end{tabular}
\end{scriptsize}
\caption{Time (in microseconds) to compute easy values to
10, 100, 1000, 10\,000 and 100\,000 digits.
Note: in case C, Arb takes 360\,s the first time at $10^5$ digits, due to
precomputing Bernoulli numbers.}
\label{tab:prectimings}
\end{center}
\end{table}

\subsection{Large parameters and argument}

\Cref{tab:ultrabenchmark} compares global performance of
confluent hypergeometric functions.
Each function is evaluated at 61 exponentially spaced arguments
$z_k = \pi 10^{k/10}, 0 \le k \le 60$
for varying precision or parameter magnitudes,
covering the convergent, asymptotic and transition regimes.

We include timings for MPFR, which implements $J_N(z)$ for $N \in \mathbb{Z}, z \in \mathbb{R}$
in floating-point arithmetic with correct rounding.
At low precision, Arb computes $J_N(z)$ about 2-3 times slower than MPFR.
This factor is explained by Arb's lack of automatic compensation for precision loss,
the ${}_0F_1(b,z)$ series evaluation not being optimized specifically for $b \in \mathbb{Z}, z \in \mathbb{R}$,
and complex arithmetic being used for the asymptotic expansion,
all of which could be addressed in the future.

\Cref{fig:timings} shows a more detailed view of test cases A (MPFR, Arb) and G (Mathematica, Arb).
The evaluation time peaks in the transition region between the convergent and asymptotic series.
Approaching the peak, the time increases smoothly as more terms are required.
With Arb, sudden jumps are visible where the precision is doubled.
By tuning the implementation of $J_0(z)$, the jumps could be smoothed out.
In test case B (not plotted), there is no cancellation, and no such jumps occur.
Mathematica is slow with large parameters (we get more than a factor $10^5$ speedup),
likely because it is conservative about using the asymptotic expansion.
Since Arb computes a rigorous error bound, the asymptotic expansion
can safely be used aggressively.

\begin{table}
\renewcommand{\arraystretch}{1.4}
\setlength{\tabcolsep}{.4em}
\begin{center}
\begin{scriptsize}
\begin{tabular}{l | l l l l | l l l l}
                                  & \multicolumn{4}{|c|}{Mathematica} & \multicolumn{4}{|c}{Arb} \\
Function $\backslash$ $N$         &  10       & 100       & 1000 & 10000 &  10 & 100 & 1000 & 10000  \\ \hline
A: $J_0(z)$              &  0.0039   &  0.020    & 3.0     & 4700      &  0.00097   &  0.0064  &  0.12   &  7.7  \\
B: $I_0(z)$             &   0.0032   &  0.012    & 1.5     & 2000  &  0.00081   &  0.0042  &  0.067  &  3.4  \\
C: $J_0(\omega z)$      &   0.0099   &  0.072    & 16     & 36000      &  0.0014    &  0.0069  &  0.16   &  11   \\
D: $K_0(z)$             &   0.0037   &  0.028    & 4.2     &  5900  &  0.0018    &  0.020   &  0.47   &  28   \\
E: $J_N(z)$            &   0.0038   &  0.0073   & 0.31  &  46  &  0.0010    &  0.0054  &  0.13   &  2.7  \\
F: $J_{Ni}(z)$         &   0.0089   &  28       & $>10^5$     &   $>10^5$  &  0.0017    &  0.0096   &  0.14   &  10   \\
G: ${}_1F_1(Ni,1+i,\omega z)$     &   0.13     &  15       & 15000     &  $>10^5$   &  0.0082    &  0.061   &  1.3    &  82    \\ \hline
           &   \multicolumn{4}{|c|}{MPFR} & \multicolumn{4}{|c}{ } \\ \hline
A: $J_0(z)$   & 0.00057  &  0.0030   &  0.24  &  42 & & & & \\
E: $J_N(z)$   & 0.00078  &  0.0021   &  0.039 &  2.7 & & & & \\
\end{tabular}
\end{scriptsize}
\caption{Time in seconds to evaluate the functions on 61 exponentially spaced points from $\pi$ to $10^6 \pi$.
In A--D, the function is computed to $N$ digits, $N = 10, 10^2, 10^3, 10^4$.
In E--G, the function is computed to 10 digits, with $N$ or $Ni$ ($i = \sqrt{-1}$) as a parameter. Here, $\omega = e^{\pi i / 3}$.}
\label{tab:ultrabenchmark}
\end{center}
\end{table}

\begin{figure}[!htb]
\centering
\includegraphics[scale=0.59]{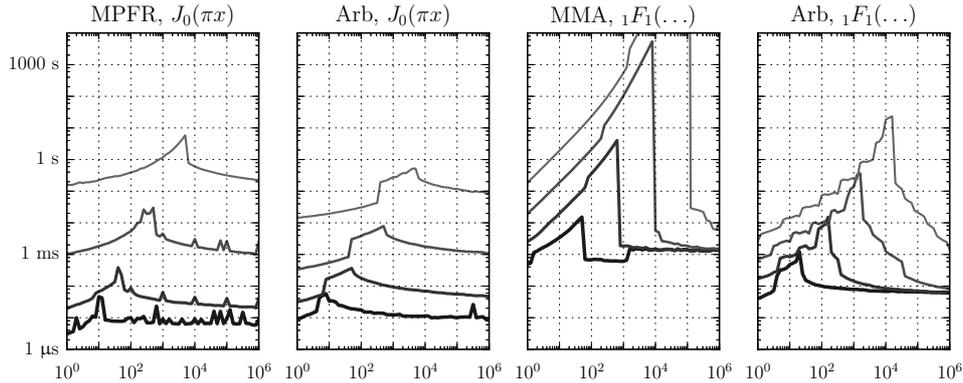}
\caption{Time as a function of $x$ to evaluate $J_0(\pi x)$ to $N$ digits and ${}_1F_1(Ni,1+i,e^{\pi i / 3} \pi x)$ to 10 digits,
for $N = 10^1, 10^2, 10^3, 10^4$ (from bottom to top).}
\label{fig:timings}
\end{figure}

\subsection{Airy function zeros}

Arb includes code for rigorously computing
all the zeros of a real analytic functions on a finite interval.
Sign tests and interval bisection are used
to find low-precision isolating intervals for all zeros,
and then as an optional stage rigorous Newton iteration
is used to refine the zeros to high precision.

Arb takes 0.42 s to isolate the 6710 zeros of
$\operatorname{Ai}(z)$ on the interval $[-1000,0]$
(performing 67\,630 function evaluations),
1.3 s to isolate and refine the zeros to 16-digit accuracy (181\,710 evaluations),
and 23 s to isolate and refine the zeros to 1000-digit accuracy
(221\,960 evaluations).
Mathematica's built-in \texttt{AiryAiZero} takes 2.4 s for
machine or 16-digit accuracy and 264 s for 1000-digit accuracy.

Note that the Arb zero-finding only uses interval evaluation
of $\operatorname{Ai}(z)$ and its derivatives; 
no a priori knowledge about the distribution of zeros
is exploited.

\subsection{Parameter derivatives}

We compute $\partial_{\nu}^n J_{\nu}(z) \vert_{\nu=1,z=\pi}$ to 100 digits,
showing timings in \cref{tab:paramderiv}.
Mathematica's \texttt{N[]} scales very poorly.
We also include timings with Maple 2016 (on the same machine as Mathematica),
using \texttt{fdiff} which implements numerical differentiation.
This performs better, but the automatic
differentiation in Arb is far superior.
In fact, Maple automatically uses parallel computation (8 cores), and its total CPU
time is several times higher than the wall time shown in~\cref{tab:paramderiv}.

\begin{table}[h!]
\renewcommand{\arraystretch}{1.1}
\setlength{\tabcolsep}{.4em}
\begin{center}
\begin{small}
\begin{tabular}{l | l l l l l l l}
Code \textbackslash $\;n$      & 1 & 2 & 5 & 10 & 100 & 1000 & 10000  \\ \hline
Mathematica     & 0.12     & 0.20    & 2600    &         &       &     &        \\
Maple (8 cores) & 0.024    & 0.039   & 4.3     & 35       & $>1$ h   &     &        \\
Arb             & 0.000088 & 0.00015 & 0.00031 & 0.00058 & 0.017 & 6.6 & 1400
\end{tabular}
\end{small}
\caption{Time (in seconds) to compute parameter $n$-th derivatives.}
\label{tab:paramderiv}
\end{center}
\end{table}

\section{Conclusion}

We have demonstrated that it is practical
to guarantee rigorous error bounds for arbitrary-precision evaluation
of a wide range of special functions, even with complex parameters.
Interval arithmetic is seen to work very well, and the methods
presented here could be exploited in other software.
The implementation in Arb is already being used in applications,
but it is a work in progress, and many details could still be optimized.
For example, recurrence relations and alternative evaluation formulas
could be incorporated to reduce cancellation,
and internal parameters such as the number of terms
and the working precision could be chosen more intelligently.
More fundamentally, we have not addressed the following important issues:

\begin{itemize}
\item Rigorous error bounds for asymptotic expansions
of the generalized hypergeometric function ${}_pF_q$ in cases
not covered by the $U$-function.
\item Efficient support for exponentially large parameter values (e.g.\ $|a| \gg 10^4$, say
in time polynomial in $\log |a|$ rather than in $|a|$), presumably
via asymptotic expansions with respect to the parameters.
\item Methods to compute tight bounds when given wide intervals as input,
apart from the simple cases where derivatives and functional
equations can be used.
\item Using machine arithmetic where possible to speed up low-precision evaluation.
\item Formal code verification, to eliminate bugs that
may slip past both human review and randomized testing.
\end{itemize}

%\section*{Acknowledgements}

\bibliographystyle{siamplain}
\bibliography{references.bib}

\end{document}